\documentclass[aps,prl,twocolumn,superscriptaddress, amsmath, amssymb, amsfonts, nofootinbib]{revtex4-1}
\def\draft{1}
\usepackage{amsthm}
\usepackage{color}

\usepackage{ifthen}
\newcommand{\YSHI}[1]{\ifthenelse{\equal{\draft}{1}}{{\color{red}{#1}}}{#1}}

\newtheorem{theorem}{Theorem}
\newtheorem{proposition}[theorem]{Proposition}
\newtheorem{definition}[theorem]{Definition}

\newtheorem{corollary}[theorem]{Corollary}

\usepackage{graphics}

\newcommand{\commentout}[1]{}

\def\Tr{\textnormal{Tr}}

\def\Dist{\textnormal{Dist}}

\begin{document}

\title{Certified randomness between mistrustful players}

\author{Carl A.~Miller}
\email{carlmi@umich.edu}
\affiliation{National Institute of Standards and Technology, Gaithersburg, MD  20899, USA, \\
Joint Center for Quantum Information and Computer Science, University of Maryland, College Park, MD  20742, USA }

\author{Yaoyun Shi}
\email{shiyy@umich.edu}
\affiliation{Dept.~of Electrical Engineering and Computer Science, University of Michigan, Ann Arbor, MI  48109, USA}

\date{\today}

\begin{abstract}
\noindent
It is known that if two players achieve a superclassical score at a nonlocal game $G$, then their
outputs are certifiably random -- that is, regardless of the strategy used by the players, a third party
will not be able to perfectly predict their outputs (even if he were given their inputs).  
We prove that for any complete-support game $G$, there is an explicit nonzero function $F_G$
such that if Alice and Bob achieve a superclassical score of $s$ at $G$, then Bob has a probability
of at most $1 - F_G ( s )$ of correctly guessing Alice's output after the game is played.
Our result implies that certifying global randomness through such games must necessarily introduce local randomness.
\end{abstract}


\maketitle

\section{Introduction}

Bell inequality violations produce \textit{certified} randomness, a  fact which is at the center of protocols for device-independent quantum
cryptography.  When two parties Alice and Bob play a nonlocal game $G$ and achieve
a score that exceeds the best classical score $\omega_c ( G )$, their output can be guaranteed
to have a certain level of min-entropy (see, e.g., Figure 2 in \cite{pam:2010}).   A natural question then arises:
is certified randomness also generated by one player from the perspective of the other player?  Specifically, 
if $A, X$ denote Alice's input and output and $Z$ denotes Bob's device, is there a certified lower bound for the conditional
min-entropy $H_{min} ( X \mid A Z)$ after the game is played?  Besides helping us understand the nature of certified randomness, this particular
kind of randomness (local randomness) has applications in mutually mistrustful cryptographic settings, where
Alice and Bob are cooperating but have different interests.

Quantifying local randomness is challenging because many of the standard tools do not apply.
Lower bounds for global randomness such as those in Figure 2 in \cite{pam:2010} are based
on constraints on the probability distribution generated by Alice and Bob and are not directly useful.
Even tools that were generated specifically for the purpose of addressing quantum side information, such as those
in our previous paper on randomness expansion \cite{MS-STOC14}, assume that the quantum information remains static during the 
generation of the randomness, and are not applicable.  But indeed, one of the reasons local randomness can appear is because 
Bob's measurement causes him to lose information that he otherwise could have retained
about Alice's system (as in, e.g., the optimal strategy for the CHSH game, where Alice and Bob
make incompatible measurements on a maximally entangled qubit-pair).

Does the generation of certified randomness always involve the generation
of \textit{local} certified randomness (i.e., randomness known to only one player)?  The answer
is not obvious: for example, in the non-signaling setting, Alice and Bob could share a PR-box\footnote{That is,
the unique $2$-part non-signaling resource whose input bits $a,b$ and output bits $x,y$ always
satisfy $x \oplus y = a \wedge b$.} which generates $1$ bit of certified randomness per use, but
no new local randomness -- Bob could perfectly guess Alice's output from his own if he were given
Alice's input. 

Motivated by the above, we prove the following result in this paper (see Theorem~\ref{robustthm} for a formal
statement).
\begin{theorem}[Informal]
\label{robustthminformal}
Let $G$ be a complete-support game\footnote{That is, a game in which each
input pair occurs with nonzero probability.}  and $\epsilon\in(0,1)$.  Suppose that Alice and Bob use a strategy for $G$
that achieves a score that is $C_G \sqrt{\epsilon}$ above the best possible classical score, where $C_G>0$ is a constant which depends on the game. Then, at the conclusion of the strategy, and given Alice's input, Bob can guess
her output with probability at most $(1- \epsilon )$.
\end{theorem}
In our previous paper \cite{MSCertifying:2016}, we studied the case in which Bob can guess
Alice's output with probability $1$.  We proved a structural theorem for such strategies which implied
that the score achieved by such strategies must be classical.  Theorem~\ref{robustthminformal} can
be seen as a robust version of Corollary 6 from \cite{MSCertifying:2016}.

Global randomness (e.g., a random variable that is possessed by Alice but completely know to Bob) is
useful for cooperative tasks such as QKD, while local randomness is useful for tasks that do not involve full mutual
trust (see the \textit{Further Directions} section of \cite{MSCertifying:2016}, where we discuss
the notion of certified erasure).  This result shows that these complementary resources tend to occur together.

Our proof proceeds as follows.
It has been observed by previous work (e.g., \cite{vidickthesis}, \cite{Manc:2015}) that if a measurement $\{ P_i \}$
on a system $D$ from a bipartite state $\rho_{DE}$ are highly predictable via measurements on $E$, then
the measurement does not disturb the reduced state by much: $\sum_i P_i \rho_D P_i \sim \rho_D$.  We include
a simplified proof of this fact (Proposition~\ref{disturbprop}).  The interesting
consequence for our purpose is that if Alice's measurements are highly predictable to Bob, then Alice can
copy our her measurement outcomes in advance, thus making her strategy approximately classical.  We take
this a step further, and show that if Bob first performs his own measurement on $E$, then the correlation of
the outcome of his measurement with $D$ is also approximately preserved by Alice's measurements (which is not necessarily true of the original entangled state $\rho_{AB}$).
This is sufficient to show that an approximately-guessable strategy yields an approximately classical strategy.

The subtleties in the proof are in establishing the error terms that arise when Alice copies out multiple measures
from her side of the state.  We note that the proof crucially requires that the game $G$ has complete support.  (If 
$G$ does not have complete support, our results imply that there must exist a pair $(a,b)$ which generates local randomness,
but this pair may not be in the support of $G$.)  An interesting
further avenue is to explore how local randomness behaves in games without complete support.

\paragraph*{Preliminaries.}

We will use the same notation as in \cite{MSCertifying:2016}, with some additions.
For any finite-dimensional Hilbert space $V$, let
$L ( V ) $ denote the vector space of linear automorphisms of $V$.
For any $M, N \in L ( V )$, let $\left< M , N \right> := \Tr [ M^* N]$.
If $S \subseteq V$ is a subspace of $V$,
let $\mathbf{P}_S \in L ( V )$ denote orthogonal projection
onto $V$.

To avoid unnecessary repetition, throughout the paper
we fix four disjoint sets $\mathcal{A, B, X, Y }$, which denote,
respectively, the first player's input alphabet, the second 
player's input alphabet, the first player's output alphabet, and
the second player's output alphabet.
A \textit{$2$-player correlation} is a vector
$(p_{ab}^{xy} )$ of nonnegative real numbers,
indexed by $(a, b, x, y ) \in \mathcal{A} \times \mathcal{B}
\times \mathcal{X} \times \mathcal{Y}$, such that
\begin{eqnarray}
\sum_{xy} p_{ab}^{xy} = 1
\end{eqnarray}
for all $(a, b) \in \mathcal{A} \times \mathcal{B}$,
and such that the quantities
\begin{eqnarray}
\label{nonsig1}
p_a^x & := & \sum_y p_{ab}^{xy} \\
\label{nonsig2}
p_b^y & := & \sum_x p_{ab}^{xy}
\end{eqnarray}
are, respectively, independent of $b$ and independent of $a$.  (These last
two equations assert non-signaling between Alice and Bob.)

A \textit{$2$-player game} is a pair $G = (q, H)$ such that
\begin{eqnarray}
q \colon \mathcal{A} \times \mathcal{B} \to [0, 1]
\end{eqnarray}
is a probability distribution and 
\begin{eqnarray}
H \colon \mathcal{A} \times \mathcal{B} \times \mathcal{X} \times \mathcal{Y} \to [0, 1 ]
\end{eqnarray}
is a function.  The function $q$ denotes the input distribution 
of the game, and $H$ denotes the scoring function.  Thus, the expected score associated
to the correlation $(p_{ab}^{xy} )$ under the game $G$ is
\begin{eqnarray}
\sum_{abxy} q ( a, b ) H ( a, b, x, y ) p_{ab}^{xy}.
\end{eqnarray}

We will extend notation by writing $q(a) = \sum_b q(a,b)$, $q(b) = \sum_a q(a,b)$.
If $q$ is such that $q ( a, b ) > 0$ for all $(a, b ) \in \mathcal{A} \times \mathcal{B}$, then
the game is said to have \textit{complete support}.   In such a case we will also write
$\mathbf{P}_q ( a \mid b ) = q(a, b) / q ( b)$.

A \textit{$2$-player strategy} is a $5$-tuple
\begin{eqnarray}
\Gamma & = & ( D, E, \{ \{ R_a^x \}_x \}_a ,
\{ \{ S_b^y \}_y \}_b , \gamma )
\end{eqnarray}
such that $D, E$ are finite-dimensional complex Hilbert spaces,
$\{ \{ R_a^x \}_x \}_a$ is a family of $\mathcal{X}$-valued
POVMs\footnote{The acronym POVM stands for positive operator-valued
measurement.} on $D$, $\{ \{ S_b^y \}_y \}_b$
is a family of $\mathcal{Y}$-valued POVMs on $E$, and $\gamma$
is a density operator on $D \otimes E$.  The \textit{second player states}
of $\Gamma$, denoted by $\rho_{ab}^{xy}$, are defined by
\begin{eqnarray}
\rho_{ab}^{xy} & = & \Tr_D \left[ \sqrt{ R_a^x \otimes
S_b^y } \gamma \sqrt{ R_a^x \otimes S_b^y } \right].
\end{eqnarray}
We also define $\rho_a^x$ by the same
expression with $S_b^y$ replaced by the identity operator.
(The states $\rho_a^x$ are the states that occur in the system
$E$ before the second player has performed a measurement.)
Let $\rho = \Tr_ D \gamma$.

We say that the strategy $\Gamma$ \textit{achieves}
the $2$-player correlation $(p_{ab}^{xy})$ if
$p_{ab}^{xy} = \Tr [ ( R_a^x \otimes S_b^y ) \gamma ]$.  
If $(p_{ab}^{xy})$ can be achieved by a strategy in which the
state $\gamma$ is separable, the we say that
$(p_{ab}^{xy})$ is a \textit{classical} correlation.
Let $\omega_c ( G )$ denote the maximum score achieved
at $G$ by classical correlations.

\begin{definition}
Let $\{ \rho_i \}_{i=1}^n$ denote a finite set of positive semidefinite
operators on a finite dimensional Hilbert space $V$.  Then, let 
\begin{eqnarray}
\Dist \{ \rho_i \} & = & \max \sum_i \Tr ( T_i \rho_i ), 
\end{eqnarray}
where the maximum is taken over all POVMs $\{ T_i \}_{i=1}^n$
on $V$.
\end{definition}

Note that if $\sum_i \Tr ( \rho_i) = 1$, and each $\rho_i$ is nonzero, then this quantity has
the following interpretation:
if Alice gives Bob a state from the set $\{ \rho_i / \Tr ( \rho_i ) \}$
at random according to the distribution $(\Tr (\rho_i ) )_i$, then $\Dist \{ \rho_i \}$
is the optimal probability that Bob can correctly guess the state.  This quantity
is well-studied (see, for example, Chapter 5 in \cite{Spehner:2014}).

\paragraph*{Main result.} Now we develop the tools to prove our main result.

\begin{definition}
Let $\Phi \colon L ( V ) \to L ( V )$ denote a completely positive trace-preserving
map over a finite-dimensional Hilbert space $V$.  Let $\beta \in L ( V )$ denote a density
operator on $V$.  Then we say that \textbf{$\Phi$ is $\epsilon$-commutative
with $\beta$ if}
\begin{eqnarray}
\left\| \Phi ( \beta ) - \beta \right\|_1 & \leq & \epsilon.
\end{eqnarray}
\end{definition}
Note that this relation obeys a natural triangle inequality:
if $\Phi_1$ is $\epsilon_1$-commutative with $\beta$, and $\Phi_2$
is $\epsilon_2$-commutative with $\beta$, then
\begin{eqnarray*}
\left\| \Phi_2 ( \Phi_1 ( \beta )) - \rho \right\|_1 & \leq &
\left\| \Phi_2 ( \Phi_1 ( \beta ) ) - \Phi_2 ( \beta) \right\|_1 + \left\|
\Phi_2 ( \beta ) - \beta \right\|_1 \\
& \leq & \left\| \Phi_1 ( \beta) - \beta \right\|_1 + \epsilon_2 \\
& \leq & \epsilon_1 + \epsilon_2.
\end{eqnarray*}

The following proposition will be an important building block.  Our proof is a significant
simplification of a method from Lemma 29 in \cite{vidickthesis}.
(See also Lemma 2 in \cite{Manc:2015} for a related result.)

\begin{proposition}
\label{disturbprop}
Let $\Lambda \in L ( A \otimes B )$ be a density operator
and $\{ F_i \}_{i=1}^n$ a projective measurement on $A$ such that the induced states
$\Lambda_i^B := \Tr_A ( F_i \Lambda )$ satisfy
\begin{eqnarray}
\Dist \{ \Lambda^B_i \} & = & 1 - \delta.
\end{eqnarray}
Then, the superoperator $X \mapsto \sum_i F_i X F_i$ is
$(2 \sqrt{\delta} + \delta )$-commutative with $\Lambda^A := \Tr_B \Lambda$.
\end{proposition}

\begin{proof}
By assumption, there exists a POVM $\{G_i \}$ on $B$ such that
\begin{eqnarray}
\sum_i \Tr [ (F_i \otimes G_i ) \Lambda ] & = & 1 - \delta.
\end{eqnarray}
By standard arguments, we can assume without loss of generality that
that $\{G_i \}$ is a projective measurement and that $\Lambda$ is pure.\footnote{
We can construct an enlargement
$B \subseteq \overline{B}$ such that $\mathbf{P}_B \overline{G}_i
\mathbf{P}_B = G_i$ for some projective measurement
$\{ \overline{G}_i \}$ on $\overline{B}$, and we can construct
an additional Hilbert space $E$ and a pure state $\overline{\Lambda}
\in L ( A \otimes B \otimes E )$ such that $\Tr_E \overline{\Lambda}
= \Lambda$.  The joint probability distribution of the measurements
$\{ F_i \}$ and $\{ \overline{G}_i \otimes I_E \}$ on $\overline{\Lambda}$
are the same as those of $\{ F_i \}$ and $\{ G_i \}$ on $\Lambda$.}

There is a linear map $M \colon \mathbb{C}^s \to \mathbb{C}^r$ 
such that $\Tr_A \Lambda = M^* M$ and $\rho = \Tr_B \Lambda = M M^*$.
Upon choosing an appropriate basis for $A$ and $B$, we
can write $M$ with a block form determined by the spans of $\{F_i \}$ and $\{ G_j \}$:
\begin{eqnarray}
M & = & \left[ \begin{array}{c|c|c|c}
M_{11} & M_{12} & \cdots & M_{1n} \\
\hline
M_{21} & M_{22} & \cdots & M_{2n} \\
\hline
\vdots & & \ddots \\
\hline
M_{n1} & M_{n2} & \cdots & M_{nn}
\end{array} \right].
\end{eqnarray}
Let
\begin{eqnarray}
\overline{M} & = & \left[ \begin{array}{c|c|c|c}
M_{11} & 0 & \cdots & 0 \\
\hline
0 & M_{22} & \cdots & 0 \\
\hline
\vdots & & \ddots \\
\hline
0 & 0 & \cdots & M_{nn}
\end{array} \right].
\end{eqnarray}
Note that the probability of obtaining outcome $F_i$ for the measurement
on $A$ and outcome $G_j$ for the measurement on $B$ is
given by the quantity $\left\| M_{ij} \right\|_2^2$, and  the probability
that the outcomes of the measurements disagree is exactly $\left\| M -
\overline{M} \right\|_2^2$.  We have
\begin{eqnarray}
\left\| M - \overline{M} \right\|_2^2 & = & \delta.
\end{eqnarray}
Additionally, we can compare $\overline{M} \overline{M}^*$ to the post-measurement
state $\sum_i F_i \rho F_i$.  The latter quantity is given by
\begin{eqnarray*}
\left[ \begin{array}{c|c|c|c}
\sum_k M_{1k} M_{1k}^* & 0 & \cdots & 0 \\
\hline
0 &  \sum_k M_{2k} M_{2k}^* & \cdots & 0 \\
\hline
\vdots & & \ddots \\
\hline
0 & 0 & \cdots & \sum_k M_{nk} M_{nk}^* \\
\end{array} \right], 
\end{eqnarray*}
and therefore the difference $(\sum_i F_i \rho F_i  - \overline{M} \overline{M}^* )$
is equal to
\begin{eqnarray*}
\left[ \begin{array}{c|c|c|c}
\sum_{k\neq 1} M_{1k} M_{1k}^* & 0 & \cdots & 0 \\
\hline
0 &  \sum_{k \neq 2} M_{2k} M_{2k}^* & \cdots & 0 \\
\hline
\vdots & & \ddots \\
\hline
0 & 0 & \cdots & \sum_{k \neq n} M_{nk} M_{nk}^*, \\
\end{array} \right] \\
\end{eqnarray*}
which is a positive semidefinite operator
whose trace is exactly $\sum_{i \neq j} \left\| 
M_{ij} \right\|^2_2 = \delta$.  Thus,
\begin{eqnarray}
\left\| \sum_i F_i \rho F_i - \overline{M} \overline{M}^* \right\|_1 
& = & \delta.
\end{eqnarray}
Therefore we have the following, using the Cauchy-Schwarz inequality:
\begin{eqnarray*}
&& \left\| \rho - \sum_i F_i \rho F_i \right\|_1 \\
& = &
\left\| M M^* - \sum_i F_i \rho F_i \right\|_1 \\
& = & \left\| M ( M - \overline{M}^* ) + (M - \overline{M}) \overline{M}^* + \overline{M} \overline{M}^* - 
\sum_i F_i \rho F_i \right\|_1 \\
& \leq & \left\| M ( M - \overline{M}^* ) \right\|_1 + \left\| (M - \overline{M}) \overline{M}^* \right\|_1 
\\ && + \left\| \overline{M} \overline{M}^* - 
\sum_i F_i \rho F_i \right\|_1  \\
& \leq & \left\| M \right\|_2 \left\|  M - \overline{M}^*  \right\|_2 + \left\| M - \overline{M} \right\|_2 \left\| \overline{M}^* \right\|_2 + \delta \\
& \leq & 1 \cdot \sqrt{\delta} + \sqrt{\delta} \cdot 1 + \delta \\
& \leq & 2 \sqrt{\delta} + \delta,
\end{eqnarray*}
as desired.
\end{proof}

\begin{corollary}
\label{classicalinfocor}
Let $\Lambda \in L ( A \otimes B \otimes C)$ be a density operator
which is classical\footnote{That is, $\Lambda = \sum_k \Lambda_k \otimes
\left| c_k \right> \left< c_k \right|$ for some orthonormal basis $\{c_1, \ldots, c_k \} 
\subseteq C$.} on $C$.  Suppose that $\{ F_i \}_{i=1}^n$ is a projective measurement on $A$ such that the induced
states $\Lambda^{BC}_i := \Tr_A ( F_i \Lambda )$ satisfy
\begin{eqnarray}
\label{assumptionrepeat}
\Dist \{ \Lambda^{BC}_i \} & = & 1 - \delta.
\end{eqnarray}
Then, the superoperator $X \mapsto \sum_i (F_i \otimes I ) X (F_i \otimes I)$ is
$(2 \sqrt{\delta} + \delta )$-commutative with $\Lambda^{AC}$.
\end{corollary}

\begin{proof}
Let $\overline{C}$ be a Hilbert space
which is isomorphic to $C$, and let $\overline{\Lambda} \in L ( A \otimes B \otimes C \otimes \overline{C})$
be the state that arises from $\Lambda$ by copying out along the standard
basis: $\left| c_i \right> \mapsto \left| c_i \overline{c_i } \right>$.  This copying
leaves the state $ABC$ unaffected, so assumption (\ref{assumptionrepeat}) still applies.
Thus by Proposition~\ref{disturbprop}, the operator
$X \mapsto \sum_i ( F_i \otimes I ) X (F_i \otimes I )$ is $(2 \sqrt{\delta} + \delta )$-commutative
with $\Lambda^{A \overline{C}}$, and the same
holds for the isomorphic state $\Lambda^{AC}$.
\end{proof}

\begin{proposition}
\label{uniformprop}
Let
\begin{eqnarray}
\Gamma & = & ( D, E, \{ \{ R_a^x \}_x \}_a , 
\{ \{ S_b^y \}_y \}_b , \gamma )
\end{eqnarray}
be a two-player strategy.
Let
\begin{eqnarray}
\delta & = & 1 - \frac{1}{| \mathcal{A} | |\mathcal{B } |} \sum_{aby} \Dist \{ \rho_{ab}^{xy} \mid x \in \mathcal{X} \}.
\end{eqnarray}
Then, there exists a classical correlation $(\overline{p}_{ab}^{xy} )$ such that
\begin{eqnarray}
\frac{1}{| \mathcal{A} | |\mathcal{B} | } \sum_{abxy} \left| p_{ab}^{xy} - \overline{p}_{ab}^{xy} \right| & \leq & 
\sqrt{ 3 \delta} \left| \mathcal{A} \right|.
\end{eqnarray}
\end{proposition}

\begin{proof}
We can assume without loss of generality that the measurements $\{ \{ R_a^x \}_x \}_a$ are all projective.
We begin with the same strategy as in the proof of  Proposition~5 in \cite{MSCertifying:2016}.
For each $a \in \mathcal{A}$, let $V_a = \mathbb{C}^\mathcal{X}$,
and let $\Phi_a \colon L ( D ) \to L ( V_a \otimes D )$
be the nondestructive measurement defined by 
\begin{eqnarray}
\Phi_a ( T ) & = & \sum_{x \in \mathcal{X}} \left| x \right> \left< x \right|
\otimes R_a^x T R_a^x.
\end{eqnarray}
Let $\Phi_a^{V_a} = \Tr_D \circ \Phi_a$ and let $\Phi_a^D = \Tr_{V_a} \circ
\Phi_a$. Likewise let $W_b = \mathbb{C}^\mathcal{Y}$ for each $b \in \mathcal{B}$, let
$\Psi_b \colon L ( D ) \to L ( W_b \otimes D )$ be the nondestructive measurement
defined by
\begin{eqnarray}
\Psi_b ( T ) & = & \sum_{y \in \mathcal{Y}} \left| y \right> \left< y \right|
\otimes \sqrt{S_b^y} T \sqrt{S_b^y}.
\end{eqnarray}
Let $\Psi_b^{W_b} = \Tr_E \circ \Psi_b$ and $\Psi_b^E = \Tr_{W_b} \circ \Psi_b$.

Assume without loss of generality that $\mathcal{A} = \{ 1, 2, \ldots, n \}$.
Let $\Lambda  \in L ( V_1 \otimes \ldots \otimes V_n \otimes D \otimes E )$
be the state that arises from applying the superoperators
$\Phi_1 \otimes I_E, \ldots, \Phi_n \otimes I_E$, in order, to $\gamma$.  Let
$(\overline{p}_{ab}^{xy} )$ be the correlation that arises
from Alice and Bob sharing the reduced state $\Lambda^{V_1 \ldots V_n E }$,
Alice obtaining her output on input $a$ from
the register $V_a$, and Bob obtaining his output from his prescribed measurements
$\{ \{S_b^y \}_y \}_b$ to $E$.  Since the state $\Lambda^{V_1 \ldots V_n E }$
is a separable state over the bipartition $(V_1 \ldots V_n \mid E )$, the correlation
$( \overline{p}_{ab}^{xy} )$ is classical.

If Alice and Bob share the measured state $(I_D \otimes \Psi_b ) (\gamma)$ partitioned
as $(D \mid E W_b)$, then the probability that
Bob can guess Alice's outcome when she measures with $\{ R_a^x \}_x$ is given by 
\begin{eqnarray}
\delta_{ab} & := &  1 - \sum_y \Dist \{ \rho_{ab}^{xy} \mid x \in \mathcal{X} \}.
\end{eqnarray}
By Corollary~\ref{classicalinfocor}, the operator $(\Phi_a^D \otimes I_{W_b} )$
is $(2 \sqrt{\delta_{ab}} + \delta_{ab} )$-commutative with 
$(I_D \otimes \Psi_b^{W_b} ) \gamma$.

We wish to compare $(p_{ab}^{xy} )$ and $(\overline{p}_{ab}^{xy} )$.  
For any $a, b$ the probability
vector $(\overline{p}_{ab}^{xy})_{xy}$ describes the joint distribution
of the registers $V_a W_b$ under the density operator
\begin{eqnarray}
((\Phi_a^{V_a} \circ \Phi^D_{a-1} \circ \Phi^D_{a-2} \circ \cdots \circ \Phi^D_1 ) \otimes \Psi^{W_b}_b) \gamma,
\end{eqnarray}
which by the previous paragraph is within trace-distance $\sum_{i=1}^{a-1}
(2 \sqrt{\delta_{ab} } + \delta_{ab} )$ from the distribution described by
$(p_{ab}^{xy})_{xy}$:
\begin{eqnarray}
(\Phi_a^{V_a}  \otimes \Psi^{W_b}_b) \gamma.
\end{eqnarray}

Thus we have the following, in which we use the Cauchy-Schwarz inequality:
\begin{eqnarray}
\sum_{abxy} \left| p_{ab}^{xy} - \overline{p}_{ab}^{xy} \right| & \leq & 
\sum_{ab} \sum_{i=1}^{a-1} ( 2 \sqrt{\delta_{ib}} + \delta_{ib} ) \\
& = & \sum_{ab} (n-a) (2 \sqrt{\delta_{ab}} + \delta_{ab} ). \\
& \leq & \sum_{ab} (n-a) 3 \sqrt{\delta_{ab}} \\
& \leq & 3 \sqrt{ \sum_{ab} (n-a)^2 } \sqrt{ \sum_{ab} \delta_{ab}}  \\
& = & 3 \sqrt{ \sum_{ab} (n-a)^2 }  \sqrt{n | \mathcal{B} | \delta} \\
& = & 3 \sqrt{ | \mathcal{B}  | \sum_{a} (n-a)^2 } \sqrt{n | \mathcal{B} | \delta} \\
& = & 3 | \mathcal{B}  | \sqrt{ \sum_{a} (n-a)^2 } \sqrt{n \delta} \\
& \leq & 3 | \mathcal{B}| \sqrt{ n^3/3 } \sqrt{n \delta},
\end{eqnarray}
which simplifies to the desired bound.
\end{proof}

Proposition~\ref{uniformprop} is useful for addressing any game $(q, H)$
where the distribution $q$ is uniform (i.e., $q(a, b) = 1/(|\mathcal{A} | | \mathcal{B} | )$.)  We prove the following theorem
which applies to more general games.

\begin{theorem}
\label{robustthm}
Let $G = (q, H)$ be a complete-support game and let 
\begin{eqnarray}
\Gamma & = & ( D, E, \{ \{ R_a^x \}_x \}_a , 
\{ \{ S_b^y \}_y \}_b , \gamma )
\end{eqnarray}
be a two-player strategy.
Let
\begin{eqnarray}
\epsilon & = & 1 - \sum_{ab} q(a, b)  \sum_y \Dist \{ \rho_{ab}^{xy} \mid x \in \mathcal{X} \}.
\end{eqnarray}
Then, the score achieved by $\Gamma$ exceeds the best classical score $\omega_c (G )$ by at most
$C_G \sqrt{\epsilon}$, where
\begin{eqnarray}
\label{cg}
C_G & = & (3/2) \sqrt{ \sum_{ab} q ( b ) \left( \mathbf{P}_q ( a \mid b ) \right)^{-1}}.
\end{eqnarray}
\end{theorem}

\begin{proof}
Define $\overline{p}_{ab}^{xy}$ and $\delta_{ab}$ as in Proposition~\ref{uniformprop}.  We have the following (again
using the Cauchy-Schwartz inequality):
\begin{eqnarray}
&& \sum_{abxy} q ( a, b ) | p_{ab}^{xy} - \overline{p}_{ab}^{xy} |
\\
& \leq & \sum_{ab} q ( a, b ) \sum_{i=1}^{a-1} (2 \sqrt{\delta_{ib}} + \delta_{ib} ) \\
& \leq & \sum_{ab} q ( a, b ) \sum_{i=1}^{a-1} 3 \sqrt{\delta_{ib}}   \\
& = & \sum_{ab} \left( \sum_{k = a+1}^n q ( k, b ) \right) 3 \sqrt{\delta_{ab}} \\
& = & \sum_{ab} \left( \frac{\sum_{k=a+1}^n q ( k , b ) }{\sqrt{q ( a, b ) } }\right) 3 \sqrt{ q ( a, b ) \delta_{ab}} \\
& \leq & 3 \sqrt{ \sum_{ab} \frac{ (\sum_{k=a+1}^n q ( k , b ) )^2}{q ( a, b ) }} \sqrt{ \sum_{ab} q ( a, b ) \delta_{ab}} \\
& \leq & 3 \sqrt{ \sum_{ab} \frac{ (\sum_{k=a+1}^n q ( k , b ) )^2}{q ( a, b ) }} \sqrt{ \epsilon } \\
& \leq & 3 \sqrt{ \sum_{ab} \frac{ q ( b )^2}{q ( a, b ) }}\sqrt{\epsilon} \\
& \leq & 2 C_G \sqrt{ \epsilon} \label{finalbound}
\end{eqnarray}
Note that for any probability vectors $\mathbf{t} = (t_1, \ldots, t_m)$ and $\mathbf{s} = (s_1, \ldots, s_m)$ and any arbitrary vector
$(u_1, \ldots, u_m) \in [0, 1]^m$, we have
\begin{eqnarray}
\sum_i u_i ( t_i - s_i ) \leq \frac{1}{2} \sum \left\| t_i - u_i \right\|.
\end{eqnarray}
Applying this fact to the probability vectors $(q(a,b)p_{ab}^{xy})_{abxy}$
and $(q(a,b)\overline{p}_{ab}^{xy} )_{abxy}$  and the 
vector $(H ( a, b, x, y ) )_{abxy}$ implies that the difference
between the score achieved by $(p_{ab}^{xy})$ and the score
achieved by $(\overline{p}_{ab}^{xy})$ is no more than half the
quantity (\ref{finalbound}), which yields the desired result.
\end{proof}

\section{Discussion}

We have given an explicit lower bound for the amount by which
any superclassical strategy for a complete support game must increase
the randomness possessed by one of the players. 
In other words, we have achieved one-shot \textit{blind}
randomness expansion.  If $G$ is a complete support game and Alice and Bob achieve score $w$, then
Bob's probability of guessing her input given her output is at most
\begin{eqnarray}
f_G ( \omega ) & = & \left\{ \begin{array}{cl}
1 - ( w - \omega_c ( G ) )^2  / C_G & \textnormal{ if } w \geq \omega_c ( G ) \\
1 & \textnormal{ otherwise,} \end{array} \right.
\end{eqnarray}
where $C_G$ denotes the constant defined in equation (\ref{cg}).  
A good next step would be to find methods to optimize this bound.
Self-testing results (e.g., \cite{mys:2012}, \cite{WuBancal:2016}) provide
strong constraints on the behavior of players for certain nonlocal games,
and may be useful for providing better bounds on Alice's unpredictability
when the expected score is close to optimal.

An interesting goal would be to prove a multi-shot version
of Theorem~\ref{robustthminformal}, e.g., a proof that Alice's outputs across multiple rounds
have high smooth min-entropy from Bob's perspective.    This would
provide a full proof of blind randomness expansion.  (One
consequence of such a proof is that it would be possible to reduce the number
of devices needed for unbounded randomness expansion.  This is known
to be possible with four spatially separated devices \cite{MS14v4}, by cross-feeding
two copies of a bounded randomness expansion protocol.  Blind randomness
expansion would allow us to reduce this number to three.)
The recent entropy accumulation theorem
\cite{Dupuis:2016} proves lower bounds on min-entropy in diverse scenarios, and it will be interesting
to see if it can be generalized to cover blind randomness expansion as well.
A possible approach is to find an appropriate chain rule for Renyi entropies
(see Section 3 of \cite{Dupuis:2016}) which applies to the blind randomness
expansion scenario.

\paragraph{Acknowledgments.} 
We are indebted to Laura Mancinska for discussions that helped
us to figure out our proof of Theorem~\ref{robustthm}, and to Jedrzej Kaniewski for background
on two-party cryptography.
This research was supported in part by US NSF
Awards 1500095, 1216729, 1526928, and 1318070.

\bibliographystyle{apsrev}
\bibliography{../quantumsec}

\begin{thebibliography}{10}
\expandafter\ifx\csname natexlab\endcsname\relax\def\natexlab#1{#1}\fi
\expandafter\ifx\csname bibnamefont\endcsname\relax
  \def\bibnamefont#1{#1}\fi
\expandafter\ifx\csname bibfnamefont\endcsname\relax
  \def\bibfnamefont#1{#1}\fi
\expandafter\ifx\csname citenamefont\endcsname\relax
  \def\citenamefont#1{#1}\fi
\expandafter\ifx\csname url\endcsname\relax
  \def\url#1{\texttt{#1}}\fi
\expandafter\ifx\csname urlprefix\endcsname\relax\def\urlprefix{URL }\fi
\providecommand{\bibinfo}[2]{#2}
\providecommand{\eprint}[2][]{\url{#2}}

\bibitem[{\citenamefont{Pironio et~al.}(2010)\citenamefont{Pironio, Ac{\'\i}n,
  Massar, Boyer de~la Giroday, Matsukevich, Maunz, Olmschenk, Hayes, Luo,
  Manning et~al.}}]{pam:2010}
\bibinfo{author}{\bibfnamefont{S.}~\bibnamefont{Pironio}},
  \bibinfo{author}{\bibfnamefont{A.}~\bibnamefont{Ac{\'\i}n}},
  \bibinfo{author}{\bibfnamefont{S.}~\bibnamefont{Massar}},
  \bibinfo{author}{\bibfnamefont{A.}~\bibnamefont{Boyer de~la Giroday}},
  \bibinfo{author}{\bibfnamefont{D.~N.} \bibnamefont{Matsukevich}},
  \bibinfo{author}{\bibfnamefont{P.}~\bibnamefont{Maunz}},
  \bibinfo{author}{\bibfnamefont{S.}~\bibnamefont{Olmschenk}},
  \bibinfo{author}{\bibfnamefont{D.}~\bibnamefont{Hayes}},
  \bibinfo{author}{\bibfnamefont{L.}~\bibnamefont{Luo}},
  \bibinfo{author}{\bibfnamefont{T.~A.} \bibnamefont{Manning}},
  \bibnamefont{et~al.}, \bibinfo{journal}{Nature}
  \textbf{\bibinfo{volume}{464}}, \bibinfo{pages}{1021} (\bibinfo{year}{2010}).

\bibitem[{\citenamefont{Miller and Shi}(2014{\natexlab{a}})}]{MS-STOC14}
\bibinfo{author}{\bibfnamefont{C.~A.} \bibnamefont{Miller}} \bibnamefont{and}
  \bibinfo{author}{\bibfnamefont{Y.}~\bibnamefont{Shi}}, in
  \emph{\bibinfo{booktitle}{Proceedings of the 46th Annual ACM Symposium on
  Theory of Computing}} (\bibinfo{year}{2014}{\natexlab{a}}), pp.
  \bibinfo{pages}{417--426}.

\bibitem[{\citenamefont{Miller and Shi}(2016)}]{MSCertifying:2016}
\bibinfo{author}{\bibfnamefont{C.~A.} \bibnamefont{Miller}} \bibnamefont{and}
  \bibinfo{author}{\bibfnamefont{Y.}~\bibnamefont{Shi}} (\bibinfo{year}{2016}),
  \bibinfo{note}{arXiv:1608.01011v1}.

\bibitem[{\citenamefont{Vidick}(2011)}]{vidickthesis}
\bibinfo{author}{\bibfnamefont{T.}~\bibnamefont{Vidick}}, Ph.D. thesis,
  \bibinfo{school}{University of California, Berkeley} (\bibinfo{year}{2011}).

\bibitem[{\citenamefont{Mancinska}(2015)}]{Manc:2015}
\bibinfo{author}{\bibfnamefont{L.}~\bibnamefont{Mancinska}}
  (\bibinfo{year}{2015}), \bibinfo{note}{arXiv:1506.07080}.

\bibitem[{\citenamefont{Spehner}(2014)}]{Spehner:2014}
\bibinfo{author}{\bibfnamefont{D.}~\bibnamefont{Spehner}}
  (\bibinfo{year}{2014}), \bibinfo{note}{arXiv:1407.3739v1}.

\bibitem[{\citenamefont{McKague et~al.}(2012)\citenamefont{McKague, Yang, and
  Scarani}}]{mys:2012}
\bibinfo{author}{\bibfnamefont{M.}~\bibnamefont{McKague}},
  \bibinfo{author}{\bibfnamefont{T.~H.} \bibnamefont{Yang}}, \bibnamefont{and}
  \bibinfo{author}{\bibfnamefont{V.}~\bibnamefont{Scarani}},
  \bibinfo{journal}{Journal of Physics A: Mathematical and Theoretical}
  \textbf{\bibinfo{volume}{45}}, \bibinfo{pages}{455304}
  (\bibinfo{year}{2012}),
  \urlprefix\url{http://stacks.iop.org/1751-8121/45/i=45/a=455304}.

\bibitem[{\citenamefont{Wu et~al.}(2016)\citenamefont{Wu, Bancal, McKague, and
  Scarani}}]{WuBancal:2016}
\bibinfo{author}{\bibfnamefont{X.}~\bibnamefont{Wu}},
  \bibinfo{author}{\bibfnamefont{J.-D.} \bibnamefont{Bancal}},
  \bibinfo{author}{\bibfnamefont{M.}~\bibnamefont{McKague}}, \bibnamefont{and}
  \bibinfo{author}{\bibfnamefont{V.}~\bibnamefont{Scarani}}
  (\bibinfo{year}{2016}), \bibinfo{note}{arXiv:1512.02074v2}.

\bibitem[{\citenamefont{Miller and Shi}(2014{\natexlab{b}})}]{MS14v4}
\bibinfo{author}{\bibfnamefont{C.~A.} \bibnamefont{Miller}} \bibnamefont{and}
  \bibinfo{author}{\bibfnamefont{Y.}~\bibnamefont{Shi}}
  (\bibinfo{year}{2014}{\natexlab{b}}), \bibinfo{note}{arXiv:1402.0489v4}.

\bibitem[{\citenamefont{Dupuis et~al.}(2016)\citenamefont{Dupuis, Fawzi, and
  Renner}}]{Dupuis:2016}
\bibinfo{author}{\bibfnamefont{F.}~\bibnamefont{Dupuis}},
  \bibinfo{author}{\bibfnamefont{O.}~\bibnamefont{Fawzi}}, \bibnamefont{and}
  \bibinfo{author}{\bibfnamefont{R.}~\bibnamefont{Renner}}
  (\bibinfo{year}{2016}), \bibinfo{note}{arXiv:1607.01796}.

\end{thebibliography}

\end{document}